\documentclass[copyright,creativecommons]{eptcs}
\providecommand{\event}{GandALF 2013} % Name of the event you are submitting to
\usepackage{breakurl}             % Not needed if you use pdflatex only.

\usepackage{graphicx}
\usepackage{pgf}
\usepackage{tikz}
\usepackage{amssymb}
\usepackage{amsmath}
\usepackage{amsthm}
\usepackage{bbm}
\usetikzlibrary{arrows,automata,shapes,snakes}
\usepackage[latin1]{inputenc}
\usepackage{verbatim} 
\usepackage{wrapfig}
\usepackage{color}

\newtheorem{mydef}{Definition}

\newtheorem{myprop}[mydef]{Proposition}

\newcommand{\states}{\textup{States}}
\newcommand{\agt}{\textup{Agt}}
\newcommand{\act}{\textup{Act}}
\newcommand{\mov}{\textup{Mov}}
\newcommand{\tab}{\textup{Tab}}
\newcommand{\prop}{\textup{Prop}}

\makeatletter
\newsavebox{\@brx}
\newcommand{\llangle}[1][]{\savebox{\@brx}{\(\m@th{#1\langle}\)}%
  \mathopen{\copy\@brx\kern-0.5\wd\@brx\usebox{\@brx}}}
\newcommand{\rrangle}[1][]{\savebox{\@brx}{\(\m@th{#1\rangle}\)}%
  \mathclose{\copy\@brx\kern-0.5\wd\@brx\usebox{\@brx}}}
\makeatother

\title{Alternating-time temporal logic with finite-memory strategies}
\author{Steen Vester
\institute{DTU Compute\\
Technical University of Denmark}
\institute{}
\email{stve@dtu.dk}
}
\def\titlerunning{Alternating-time temporal logic with finite-memory strategies}
\def\authorrunning{Steen Vester}
\begin{document}
\maketitle

\begin{abstract}

Model-checking the alternating-time temporal logics $ATL$ and $ATL^*$ with incomplete information is undecidable for perfect recall semantics. However, when restricting to memoryless strategies the model-checking problem becomes decidable. In this paper we consider two other types of semantics based on finite-memory strategies. One where the memory size allowed is bounded and one where the memory size is unbounded (but must be finite). This is motivated by the high complexity of model-checking with perfect recall semantics and the severe limitations of memoryless strategies. We show that both types of semantics introduced are different from perfect recall and memoryless semantics and next focus on the decidability and complexity of model-checking in both complete and incomplete information games for $ATL/ATL^*$. In particular, we show that the complexity of model-checking with bounded-memory semantics is $\Delta_2^p$-complete for $ATL$ and $PSPACE$-complete for $ATL^*$ in incomplete information games just as in the memoryless case. We also present a proof that $ATL$ and $ATL^*$ model-checking is undecidable for $n \ge 3$ players with finite-memory semantics in incomplete information games.

\end{abstract}

\section{Introduction}

The alternating-time temporal logics $ATL$ and $ATL^*$ have been studied with perfect recall semantics and memoryless semantics in both complete and incomplete information concurrent game structures \cite{AHK97, AHK02, Sch04}. The model-checking problems for these logics have applications in verification and synthesis of computing systems in which different entities interact. The complexity of model-checking with perfect recall semantics, where players are allowed to use an infinite amount of memory, is very high in some cases and even undecidable in the case of $ATL$ \cite{AHK02, DT11} with incomplete information. On the other hand, model-checking with memoryless semantics, where players are not allowed to use any memory about the history of a game, is decidable and has a much lower complexity \cite{Sch04}. The drawback is that there are many games where winning strategies exist for some coalition, but where no memoryless winning strategies exist. In this paper, we focus on the tradeoff between complexity and strategic ability with respect to the memory available to the players. Instead of considering the extreme cases of memoryless strategies and infinite memory strategies we look at finite-memory strategies as an intermediate case of the two. The motivation is the possibility to solve more games than with memoryless strategies, but without the cost that comes with infinite memory.

We introduce two new types of semantics called bounded-memory semantics and finite-memory semantics respectively. For bounded-memory semantics there is a bound on the amount of memory available to the players, whereas for finite-memory semantics players can use any finite amount of memory. We will study the expressiveness of these new types of semantics compared to memoryless and perfect recall semantics in $ATL$ and $ATL^*$ with both complete and incomplete information. Afterwards we focus on the complexity and decidability of the model-checking problem for the different cases.

Our approach have similarities with the work done in \cite{Sch04}, \cite{BLLM09} and \cite{AW09}. It is a natural extension of the framework used in \cite{Sch04} where memoryless semantics and perfect recall semantics are considered. In \cite{BLLM09} $ATL/ATL^*$ with bounded-memory semantics and strategy context is introduced for complete information games, where bounded-memory strategies are defined essentially in the same way as here. However, their use of strategy context makes the problems and algorithms considered different from ours. In \cite{AW09} a version with bounded-recall is considered where agents can only remember the last $m$ states of the play. This contrasts our approach where the players can decide what to store in the memory about the past.

\section{Concurrent game structures}

A concurrent game is played on a finite graph by a finite number of players, where the players interact by moving a token between different states along the edges of the graph. The game is played an infinite number of rounds where each round is played by letting every player independently and concurrently choose an action. The combination of actions chosen by the players along with the current state uniquely determines the successor state of the game. More formally,
\\
\begin{mydef}

A concurrent game structure (CGS) with $n$ players

$$\mathcal{G} = (\states, \agt, \act, \mov, \tab)$$
consists of

\begin{itemize}

\item \states \hspace{0.04cm} - A finite non-empty set of states

\item $\agt = \{1,...,n\}$ - A finite non-empty set of players

\item \act \hspace{0.04cm} - A finite non-empty set of actions

\item $\mov: \states \times \agt \rightarrow 2^\act \setminus \{\emptyset\}$ - A function specifying the legal actions at a given state of a given player

\item $\tab: \states \times \act^n \rightarrow \states$ - A transition function defined for each $(a_1,...,a_n) \in \act^n$ and state $s$ such that $a_j \in \mov(s,j)$ for $1 \le j \le n$

\end{itemize}

\end{mydef}

Unless otherwise noted, we implicitly assume from now on that the players in a game are named $1,...,n$ where $n = |\agt|$. Note that every player must have at least one legal action in each state. The transition function $\tab$ is defined for each state and all legal tuples of actions in that state. We also refer to such legal tuples of actions as moves. To add meaning to concurrent game structures we introduce the concept of a concurrent game model which consists of a concurrent game structure as well as a labeling of the states in the structure with propositions from some fixed, finite set $\textup{Prop}$ of proposition symbols.
\\
\begin{mydef}

A concurrent game model (CGM) is a pair $(\mathcal{G}, \pi)$ where $\mathcal{G}$ is a concurrent game structure and $\pi: \states \rightarrow \mathcal{P}(\prop)$ is a labeling function.

\end{mydef}

An example of a CGM can be seen in Figure \ref{fig:gamem1}, where the states are drawn as nodes. Transitions are drawn as edges between nodes such that there is an edge from $s$ to $s'$ labeled with the move $(a_1,...,a_n)$ if $\tab(s,(a_1,...,a_n)) = s'$. The states are labelled with propositions from the set $\prop = \{p,q\}$ in the figure.

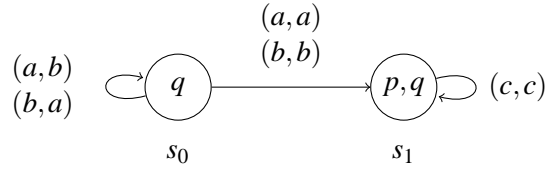
\begin{figure}[here]
\begin{center}
\begin{tikzpicture}

\tikzstyle{every node}=[ellipse, draw=black,
                        inner sep=0pt, minimum width=25pt, minimum height=25pt]

\draw (2,5) node [label=below:$s_0$] (s0)	{$q$};
\draw (5,5) node [label=below:$s_1$] (s1)	{$p,q$};

\path[->] (s0) edge node [above, draw=none] {$\begin{array}{c}
 (a,a) \\
 (b,b) \\
       \end{array}$ } (s1);
\path[->] (s1) edge [loop right] node [right, draw=none] {$(c,c)$ } (s1);

\path[->] (s0) edge [loop left] node [left, draw=none] {$\begin{array}{c}
 (a,b) \\
 (b,a) \\
       \end{array}$ } (s0);

\end{tikzpicture}
\end{center}
\caption{CGM $\mathcal{M}$}
\label{fig:gamem1}

\end{figure}

We define an incomplete information concurrent game structure as a CGS where each player $j$ has an equivalence relation $\sim_j$ on the set of states. The intuitive meaning is that $s \sim_j s'$ if player $j$ cannot distinguish between the states $s$ and $s'$.
\\
\begin{mydef}

A concurrent game structure with incomplete information (iCGS) with $n$ players is a tuple

$$\mathcal{G} = (\states, \agt, \act, \mov, \tab, (\sim_j)_{1 \le j \le n})$$
where 

\begin{itemize}

\item $(\states, \agt, \act, \mov, \tab)$ is a CGS

\item $\sim_j \subseteq \states \times \states$ is an equivalence relation for all $1 \le j \le n$

\item If $s \sim_j s'$ then $\mov(s,j) = \mov(s',j)$ for all $s,s' \in \states$ and all $j \in \agt$

\end{itemize}

\end{mydef}

Note that we require the set of actions available to a player in two indistinguishable states to be the same. We extend the notion to concurrent game models with incomplete information in the natural way.
\\
\begin{mydef}

A concurrent game model with incomplete information (iCGM) is a pair $(\mathcal{G}, \pi)$ where $\mathcal{G}$ is an iCGS and $\pi: \states \rightarrow 2^{\prop}$ is a labeling function.

\end{mydef}

For each player $j$, the relation $\sim_j$ induces a set $[\cdot]_j$ of equivalence classes of states. We denote by $[s]_j$ the class that state $s$ belongs to for player $j$. These classes are refered to as the observation sets of player $j$. Since the set of legal actions of player $j$ is required to be the same in states from the same observation set, we can define $\mov([s]_j,j) = \mov(s,j)$ for all states $s$. Note that the concepts of iCGS and iCGM generalize CGS and CGM respectively, since they are the special cases where $\sim_j$ is the identity relation for all players $j$.

\section{Outcomes, histories and strategies}

Let $\mathcal{G} = (\states, \agt, \act, \mov, \tab)$ be a CGS with $n$ players. An outcome (or play) of a concurrent game is an infinite sequence of states in the game structure that corresponds to an infinite sequence of legal moves. Formally, the set of outcomes $\textup{Out}_\mathcal{G}(s)$ of $\mathcal{G}$ from $s \in \states$ is defined as

$$\textup{Out}_\mathcal{G}(s) = \{\rho_0 \rho_1 ... \in \states^\omega \mid \rho_0 = s \wedge \forall j \ge 0 . \exists m \in \act^n . \tab(\rho_j,m) = \rho_{j+1}\}$$

$\textup{Out}_\mathcal{G} = \bigcup_{s \in \states} \textup{Out}_\mathcal{G}(s)$ is the set of all outcomes of $\mathcal{G}$. A history of a concurrent game is a non-empty, finite prefix of an outcome. The set $\textup{Hist}_\mathcal{G}(s)$ of histories of $\mathcal{G}$ from $s \in \states$ is defined as

$$\textup{Hist}_\mathcal{G}(s) = \{\rho_0 \rho_1 ... \rho_k \in \states^+ \mid \rho_0 = s \wedge \exists \rho' \in \textup{Out}_\mathcal{G}(s) . \forall 0 \le j \le k . \rho_j = \rho'_j \}$$

$\textup{Hist}_\mathcal{G} = \bigcup_{s \in \states} \textup{Hist}_\mathcal{G}(s)$ is the set of all histories of $\mathcal{G}$. For a (finite or infinite) sequence $\rho$ of states we write $\rho_0$ for the first state, $\rho_j$ for the $(j+1)$th state. $\rho_{\le j}$ is the prefix $\rho_0 \rho_1 ... \rho_j$ of $\rho$ and $\rho_{\ge j}$ is the suffix $\rho_{j} \rho_{j+1} ...$ of $\rho$. When $\rho = \rho_0 ... \rho_k$ is a finite sequence we denote the length of $\rho$ by $|\rho| = k$ and write $\textup{last}(\rho) = \rho_k$.

For a given CGS $\mathcal{G} = (\states, \agt, \act, \mov, \tab)$ we define a strategy for player $j$ as a mapping $\sigma_j: \textup{Hist}_\mathcal{G} \rightarrow \act$ such that for all $h \in \textup{Hist}_\mathcal{G}$ we have $\sigma_j(h) \in \mov(\textup{last}(h),j)$. Thus, a strategy for player $j$ maps any given history to an action that is legal for player $j$ in the final state of the history. We will also refer to these strategies as perfect recall strategies or infinite-memory strategies, since a player using such a strategy can use the entire history of a play up to the decision point to choose his next action. A memoryless (positional, no recall) strategy for player $j$ is a strategy $\sigma_j$ such that for all $h, h' \in \textup{Hist}_\mathcal{G}$ with $\textup{last}(h) = \textup{last}(h')$ we have $\sigma_j(h) = \sigma_j(h')$. It is called a memoryless strategy since the player is only using the last state of the history to decide on his action. We denote by $\textup{Strat}^R_j$ the set of perfect recall strategies for player $j$ and by $\textup{Strat}^r_j$ the set of memoryless strategies for player $j$. We write $\textup{Out}_\mathcal{G}(s,\sigma)$ for a strategy $\sigma = (\sigma_a)_{a \in \agt}$ for coalition $A$ and a state $s$ to denote the set of possible outcomes from state $s$ when players in coalition $A$ play according to $\sigma$.

Next, we define finite-memory strategies in which a player is only allowed to store a finite amount of memory of the history of the game. He can then combine his memory with the current state of the game to choose an action. To model a strategy with finite memory we use a deterministic finite-state transducer (DFST). A DFST is a 6-tuple $(M,m_0,\Sigma,\Gamma,T,G)$ where $M$ is a finite, non-empty set of states, $m_0$ is the initial state, $\Sigma$ is the input alphabet, $\Gamma$ is the output alphabet, $T: M \times \Sigma \rightarrow M$ is the transition function and $G: M \times \Sigma \rightarrow \Gamma$ is the output function. The set of states of the DFST are the possible values of the internal memory of the strategy. We will also call these memory states. The initial state corresponds to the initial memory value. The input symbols are the states of the game and the set of output symbols is the set of actions of the game. In each round of the game the DFST reads a state of the game. Then it updates its memory based on the current memory value and the input state and performs an action based on the current memory value and the input state. More formally, we say that a strategy $\sigma_j$ for player $j$ is a finite-memory strategy if there exists a DFST $A = (M,m_0,\states,\act,T,G)$ such that for all $h \in \textup{Hist}_\mathcal{G}$ we have

$$\sigma_j(h) = G(\mathcal{T}(m_0,h_{\le |h| - 1}),\textup{last}(h))$$

where $\mathcal{T}$ is defined recursively by $\mathcal{T}(m,\rho) = T(m,\rho_0)$ for any memory state $m$ and any history $\rho$ with $|\rho| = 0$ and $\mathcal{T}(m,\rho) = T(\mathcal{T}(m,\rho_{\leq |\rho| - 1}),\textup{last}(\rho))$ for any memory state $m$ and any history $\rho$ with $|\rho| \ge 1$. Intuitively $\mathcal{T}$ is the function that repeatedly applies the transition function $T$ on a sequence of inputs to calculate the memory state after a given history. We call $\mathcal{T}$ the repeated transition function. We say that $\sigma_j$ is a $k$-memory strategy if the number of states of the DFST is $k$. We also say that the strategy $\sigma_j$ is represented by the DFST $A$. We denote the set of finite-memory strategies for player $j$ by $\textup{Strat}_j^F$ and the set of $k$-memory strategies for player $j$ by $\textup{Strat}_j^{F_k}$. Thus, $\textup{Strat}_j^F = \bigcup_{k \ge 1} \textup{Strat}_j^{F_k}$. In addition, we have that the memoryless strategies are exactly the finite-memory strategies with one memory state, i.e. $\textup{Strat}_j^{F_1} = \textup{Strat}_j^r$.

Next, we generalize the notions of strategies to incomplete information games by defining them on observation histories rather than on histories, since players observe sequences of observation sets during the play rather than sequences of states. We define the set $\textup{Hist}_\mathcal{G}^j$ of observation histories for player $j$ in iCGS $ \mathcal{G}$ as

$$\textup{Hist}_\mathcal{G}^j = \{[s_0]_j [s_1]_j ... [s_k]_j \mid s_0 s_1 ... s_k \in \textup{Hist}_\mathcal{G}\}$$

For each player, a given history induces a particular observation history which is observed by the player. Then, strategies are defined as mappings from observation histories to actions, memoryless strategies are strategies where the same action is chosen for any observation history ending with the same observation set and finite-memory strategies are represented by DFSTs where the input symbols are observation sets rather than states of the game. Note that the definitions coincide for complete information games.

\section{$ATL/ATL^*$ with finite-memory and bounded-memory semantics}

The alternating-time temporal logics $ATL$ and $ATL^*$ generalize the computation tree logics $CTL$ and $CTL^*$ with the strategic operator $\llangle A \rrangle \varphi$ which expresses that coalition $A$ has a strategy to ensure the property $\varphi$. For a fixed, finite set $\agt$ of agents and finite set $\textup{Prop}$ of proposition symbols the $ATL^*$ formulas are constructed from the following grammar

$$\varphi ::= p \mid \neg \varphi_1 \mid \varphi_1 \vee \varphi_2 \mid \textbf{X} \varphi_1 \mid \varphi_1 \textbf{U} \varphi_2 \mid \llangle A \rrangle \varphi_1 $$
where $p \in \textup{Prop}$, $\varphi_1, \varphi_2$ are $ATL^*$ formulas and $A \subseteq \agt$ is a coalition of agents. The connectives $\wedge$,  $\rightarrow$, $\Leftrightarrow$, $\textbf{G}$ and $\textbf{F}$ are defined in the usual way. The universal path quantifier $\textbf{A}$ of computation tree logic can be defined as $\llangle \emptyset \rrangle$. $ATL$ is the subset of $ATL^*$ defined by the following grammar

$$\varphi ::= p \mid \neg \varphi_1 \mid \varphi_1 \vee \varphi_2 \mid \llangle A \rrangle \textbf{X} \varphi_1  \mid \llangle A \rrangle \textbf{G} \varphi_1 \mid \llangle A \rrangle (\varphi_1 \textbf{U} \varphi_2) $$
where $p \in \textup{Prop}$, $\varphi_1, \varphi_2$ are $ATL$ formulas and $A \subseteq \agt$ is a coalition of agents.

We distinguish between state formulas and path formulas, which are evaluated on states and paths of a game respectively. The state formulas are defined as follows

\begin{itemize}

\item $p$ is a state formula if $p \in \textup{Prop}$

\item If $\varphi_1$ and $\varphi_2$ are state formulas, then $\neg \varphi_1$ and $ \varphi_1 \vee \varphi_2$ are state formulas

\item If $\varphi_1$ is an $ATL^*$ formula and $A \subseteq \agt$, then $\llangle A \rrangle \varphi_1$ is a state formula

\end{itemize}
All $ATL^*$ formulas are path formulas. Note that all $ATL$ formulas are state formulas.

In \cite{Sch04} $ATL$ and $ATL^*$ are defined with different semantics based on (1) whether the game is with complete or incomplete information (2) whether perfect recall strategies or only memoryless strategies are allowed. Here $i$ and $I$ are used to denote incomplete and complete information respectively. $r$ and $R$ are used to denote memoryless and perfect recall strategies respectively. We extend this framework by considering finite-memory semantics where only finite-memory strategies are allowed and denote this by $F$. In addition we extend it with an infinite hierarchy of bounded-memory semantics, where $F_k$ for $k \ge 1$ denotes that only $k$-memory strategies are allowed. We denote the satisfaction relations $\models_{XY}$ where $X \in \{i,I\}$ and $Y \in \{r,F_1,F_2,...,F,R\}$. We will also write $ATL_{XY}$ and $ATL^*_{XY}$ to denote the logics obtained with the different types of semantics.

The semantics of formulas in alternating-time temporal logic is given with respect to a fixed CGM $\mathcal{M} = (\mathcal{G}, \pi)$ where the players that appear in the formulas must appear in $\mathcal{G}$ and the propositions present in the formulas are in $\textup{Prop}$. For state formulas we define for all CGMs $\mathcal{M} = (\mathcal{G}, \pi)$, all states $s$, all propositions $p \in \textup{Prop}$, all state formulas $\varphi_1$ and $\varphi_2$, all path formulas $\varphi_3$, all coalitions $A \in \agt$ and all $Y \in \{r,F_1,F_2,...,F,R\}$
\\

\begin{tabular}{l l}
$\mathcal{M},s \models_{IY} p $ & if $p \in \pi(s)$ \\
$\mathcal{M},s \models_{IY} \neg \varphi_1 $ & if $\mathcal{M},s \not\models_{IY} \varphi_1$ \\
$\mathcal{M},s \models_{IY} \varphi_1 \vee \varphi_2 $ & if $\mathcal{M},s \models_{IY} \varphi_1$ or $\mathcal{M},s \models_{IY} \varphi_2$ \\
$\mathcal{M},s \models_{IY} \llangle A \rrangle \varphi_3 $ & if there exist strategies $(\sigma_A)_{a \in A} \in \prod_{a \in A} \textup{Strat}^Y_a$ such that  \\
 & $\forall \rho \in \textup{Out}_\mathcal{G}(s,\sigma_A) . \mathcal{M},\rho \models_{IY} \varphi_3$\\

\end{tabular}
\\

For path formulas we define for all CGMs $\mathcal{M} = (\mathcal{G}, \pi)$, all paths $\rho$, all propositions $p \in \textup{Prop}$, all state formulas $\varphi_1$, all path formulas $\varphi_2$ and $\varphi_3$, all coalitions $A \in \agt$ and all $Y \in \{r,F_1,F_2,...,F,R\}$
\\

\begin{tabular}{l l}
$\mathcal{M},\rho \models_{IY} \varphi_1 $ & if $\mathcal{M}, \rho_0 \models_{IY} \varphi_1$ \\
$\mathcal{M},\rho \models_{IY} \neg \varphi_2 $ & if $\mathcal{M},\rho \not\models_{IY} \varphi_2$ \\
$\mathcal{M},\rho \models_{IY} \varphi_2 \vee \varphi_3 $ & if $\mathcal{M},\rho \models_{IY} \varphi_2$ or $\mathcal{M}, \rho \models_{IY} \varphi_3$ \\
$\mathcal{M},\rho \models_{IY} \textbf{X} \varphi_2 $ & if $\mathcal{M},\rho_{\ge 1}\models_{IY} \varphi_2$ \\
$\mathcal{M},\rho \models_{IY} \varphi_2 \textbf{U}\varphi_3 $ & if $\exists k . \mathcal{M},\rho_{\ge k}\models_{IY} \varphi_3$ and $\forall j < k. \mathcal{M}, \rho_{\ge j} \models_{IY} \varphi_2$ \\

\end{tabular}
\\

For iCGMs the semantics are defined similarly, but for $\llangle A \rrangle \varphi$ to be true in state $s$ the coalition $A$ must have a strategy to make sure $\varphi$ is satisfied in all plays starting in states that are indistinguishable from $s$ to one of the players in $A$. Now, for state formulas we define for all iCGMs $\mathcal{M} = (\mathcal{G}, \pi)$, all states $s$, all propositions $p \in \textup{Prop}$, all state formulas $\varphi_1$ and $\varphi_2$, all path formulas $\varphi_3$, all coalitions $A \in \agt$ and all $Y \in \{r,F_1,F_2,...,F,R\}$
\\

\begin{tabular}{l l}
$\mathcal{M},s \models_{iY} p $ & if $p \in \pi(s)$ \\
$\mathcal{M},s \models_{iY} \neg \varphi_1 $ & if $\mathcal{M},s \not\models_{iY} \varphi_1$ \\
$\mathcal{M},s \models_{iY} \varphi_1 \vee \varphi_2 $ & if $\mathcal{M},s \models_{iY} \varphi_1$ or $\mathcal{M},s \models_{iY} \varphi_2$ \\
$\mathcal{M},s \models_{iY} \llangle A \rrangle \varphi_3 $ & if there exist strategies $(\sigma_A)_{a \in A} \in \prod_{a \in A} \textup{Strat}^Y_a$ such that  \\
 & for every $a \in A$, every $s' \sim_a s$ and every $\rho \in \textup{Out}_\mathcal{G}(s',\sigma_A) $\\
 & we have $\mathcal{M},\rho \models_{iY} \varphi_3$\\

\end{tabular}
\\

For path formulas we define for all iCGMs $\mathcal{M} = (\mathcal{G}, \pi)$, all paths $\rho$, all propositions $p \in \textup{Prop}$, all state formulas $\varphi_1$, all path formulas $\varphi_2$ and $\varphi_3$, all coalitions $A \in \agt$ and all $Y \in \{r,F_1,F_2,...,F,R\}$
\\

\begin{tabular}{l l}
$\mathcal{M},\rho \models_{iY} \varphi_1 $ & if $\mathcal{M}, \rho_0 \models_{iY} \varphi_1$ \\
$\mathcal{M},\rho \models_{iY} \neg \varphi_2 $ & if $\mathcal{M},\rho \not\models_{iY} \varphi_2$ \\
$\mathcal{M},\rho \models_{iY} \varphi_2 \vee \varphi_3 $ & if $\mathcal{M},\rho \models_{iY} \varphi_2$ or $\mathcal{M}, \rho \models_{iY} \varphi_3$ \\
$\mathcal{M},\rho \models_{iY} \textbf{X} \varphi_2 $ & if $\mathcal{M},\rho_{\ge 1}\models_{iY} \varphi_2$ \\
$\mathcal{M},\rho \models_{iY} \varphi_2 \textbf{U}\varphi_3 $ & if $\exists k . \mathcal{M},\rho_{\ge k}\models_{iY} \varphi_3$ and $\forall j < k. \mathcal{M}, \rho_{\ge j} \models_{iY} \varphi_2$ \\

\end{tabular}
\\

We will occasionally write $\models_{XY}^{L}$ to emphasize that the semantics is for the logic $L$, but omit it when the logic is clear from the context as above.

\section{Expressiveness}
\label{sec:exp}
With the new types of semantics introduced we are interested in when the new types of semantics are different and when they are equivalent. For instance, in \cite{Sch04} it was noted that $\models_{Ir}$ and $\models_{IR}$ are equivalent for $ATL$, but not $ATL^*$. We do a similar comparison for the different kinds of semantics in order to understand the capabilities of different amounts of memory in different games. In addition, since there is equivalence in some cases this gives us fewer different cases to solve when considering the model-checking problem. We start by looking only at formulas of the form $\llangle A \rrangle \varphi$ where $A \subseteq \agt$ and $\varphi$ is an LTL formula. Denote the fragments of $ATL$ and $ATL^*$ restricted to this kind of formulas by $ATL_0$ and $ATL^*_0$ respectively. A nice property of these fragments is the following proposition, which tells us that to have equivalence of semantics for two types of memory for either $ATL$ or $ATL^*$ it is sufficient to consider the fragments $ATL_0$ and $ATL^*_0$ respectively.
\\

\begin{myprop}
\label{prop:prop1}
For $X \in \{i,I\}$ and $Y,Z \in \{r,F_1,F_2,...,F,R\}$ we have

\begin{enumerate}
\item $\models_{XY}^{ATL} \hspace{0.2cm} = \hspace{0.2cm} \models_{XZ}^{ATL} \textup{ if and only if } \models_{XY}^{ATL_0} \hspace{0.2cm} = \hspace{0.2cm} \models_{XZ}^{ATL_0}$

\item $\models_{XY}^{ATL^*} \hspace{0.2cm} = \hspace{0.2cm} \models_{XZ}^{ATL^*} \textup{ if and only if } \models_{XY}^{ATL^*_0} \hspace{0.2cm} = \hspace{0.2cm} \models_{XZ}^{ATL^*_0}$

\end{enumerate}
\end{myprop}

\begin{proof}

We treat both cases simultaneously and let $L \in \{ATL, ATL^*\}$. $(\Rightarrow)$ The first direction is trivial, since the set of $L_0$ formulas is included in the set of $L$ formulas. $(\Leftarrow)$ For the second direction suppose $\models_{XY}^{L_0} \hspace{0.2cm} = \hspace{0.2cm} \models_{XZ}^{L_0}$. Let $\mathcal{M} = (\mathcal{G}, \pi)$ be an (i)CGM over the set $\textup{Prop}$ of proposition symbols. Let $\varphi$ be an arbitrary formula from $L$ that contains $k$ strategy quantifiers. Let $\varphi = \varphi_0$ and $\pi = \pi_0$. We transform $\varphi_0$ and $\pi_0$ in $k$ rounds, in each round $1 \le j \le k$ the innermost subformula $\varphi'$ of $\varphi_{j-1}$ with a strategy quantifier as main connective is replaced by a new proposion $p_j \not\in \textup{Prop}$ to obtain $\varphi_j$. The labeling function is extended such that for all states $s$ we have

$$
\pi_j(s) = \left\{ \begin{array}{rl}
 \pi_{j-1}(s) \cup \{p_j\} &\mbox{ if $(\mathcal{G}, \pi_{j-1}), s \models^{L_0}_{XY} \varphi'$} \\
 \pi_{j-1}(s) &\mbox{ otherwise}
       \end{array} \right.
$$

Note that because of our initial assumption we have $(\mathcal{G}, \pi_{j-1}), s \models_{XY} \varphi'$ if and only if $(\mathcal{G}, \pi_{j-1}), s \models_{XZ} \varphi'$ since $\varphi'$ is an $L_0$ formula. Therefore, for each $j$ and all paths $\rho$ we also have
\\

$(\mathcal{G}, \pi_{j-1}), \rho \models_{XY} \varphi_{j-1}$ if and only if $(\mathcal{G}, \pi_j), \rho \models_{XY} \varphi_j$ and
\\

$(\mathcal{G}, \pi_{j-1}), \rho \models_{XZ} \varphi_{j-1}$ if and only if $(\mathcal{G}, \pi_j), \rho \models_{XZ} \varphi_j$
\\

In particular, $\varphi_k$ is an $LTL$ formula and therefore for all $\rho$ we have $(\mathcal{G}, \pi_k), \rho \models_{XY} \varphi_k$ if and only if $(\mathcal{G}, \pi_k), \rho \models_{XZ} \varphi_k$. Together with the above we get for all $\rho$ that
\\

$(\mathcal{G}, \pi_0), \rho \models_{XY} \varphi_0$ \hspace{0.2cm} iff \hspace{0.2cm} $(\mathcal{G}, \pi_1), \rho \models_{XY} \varphi_1$ \hspace{0.2cm} iff \hspace{0.2cm} $...$ \hspace{0.2cm} iff \hspace{0.2cm} $(\mathcal{G}, \pi_k), \rho \models_{XY} \varphi_k$ \hspace{0.2cm} iff\\

$(\mathcal{G}, \pi_k), \rho \models_{XZ} \varphi_k$ \hspace{0.2cm} iff \hspace{0.2cm} $...$ \hspace{0.2cm} iff \hspace{0.2cm} $(\mathcal{G}, \pi_1), \rho \models_{XZ} \varphi_1$ \hspace{0.2cm} iff \hspace{0.2cm} $(\mathcal{G}, \pi_0), \rho \models_{XZ} \varphi_0$\\

Thus, $\models_{XY}^L \hspace{0.2cm} = \hspace{0.2cm} \models_{XZ}^L$ since $\varphi$ and $\mathcal{M}$ was chosen arbitrarily.

\end{proof}

The relations between different types of semantics presented in Figure \ref{fig:tabexpress} provide insights about the need of memory for winning strategies in games with various amounts of information and types of $LTL$ objectives that can be specified in $ATL_0/ATL_0^*$. In addition, according to Proposition \ref{prop:prop1} the cases of equivalence in Figure \ref{fig:tabexpress} are exactly the cases of equivalence for the full $ATL/ATL^*$. We will use the rest of this section to prove the results of this table.

\begin{figure}[here]
\begin{center}
\renewcommand{\arraystretch}{1.5}
\begin{tabular}{l | l l l l l l l l l l l}

\multicolumn{1}{c|}{Logic} & \multicolumn{11}{c}{Expressiveness}\\
\hline
$ATL_0$ w. complete info & $\models^{ATL_0}_{Ir}$ & $=$ & $ \models^{ATL_0}_{IF_2} $ & $= $ & $\models^{ATL_0}_{IF_3} $ & $ = $ & $ ... $ & $ = $ & $ \models ^{ATL_0}_{IF} $ & $ = $ & $ \models^{ATL_0}_{IR}$ \\
\hline
$ATL_0$ w. incomplete info & $\models^{ATL_0}_{ir} $ & $\subset$ & $\models^{ATL_0}_{iF_2}$ & $\subset$ & $\models^{ATL_0}_{iF_3} $ & $ \subset$ & $ ... $ & $\subset$ & $\models^{ATL_0}_{iF}$ & $\subset$ & $\models^{ATL_0}_{iR}$\\
\hline
$ATL_0^*$ w. complete info & $\models^{ATL_0^*}_{Ir}$ & $\subset$ & $\models^{ATL_0^*}_{IF_2}$ & $\subset$ & $\models^{ATL_0^*}_{IF_3}$ & $\subset$ & $...$ & $\subset$ & $\models^{ATL_0^*}_{IF}$ & $=$ & $\models^{ATL_0^*}_{IR}$\\
\hline
$ATL_0^*$ w. incomplete info & $\models^{ATL_0^*}_{ir}$ & $\subset$ & $\models^{ATL_0^*}_{iF_2}$ & $\subset$ & $\models^{ATL_0^*}_{iF_3}$ & $\subset$ & $...$ & $\subset$ & $\models^{ATL_0^*}_{iF}$ & $\subset$ & $\models^{ATL_0^*}_{iR}$\\
\hline

\end{tabular}
\end{center}

\caption{Relations between the different types of semantics}
\label{fig:tabexpress}
\end{figure}

\subsection{Complete information games}

For complete information games, the question of whether a (memoryless/finite-memory/perfect recall) winning strategy exists for a coalition $A$ can be reduced to the question of whether a (memoryless/finite-memory/perfect recall) winning strategy exists for player 1 in a two-player turn-based game. The idea is to let player 1 control coalition $A$ and let player 2 control coalition $\agt \setminus A$ and give player 2 information about the action of player 1 before he has to choose in each round of the game in order to make it turn-based. Since $ATL_0$ can only be used to express reachability ($\llangle A \rrangle \varphi_1 \textbf{U} \varphi_2$), safety ($\llangle A \rrangle \textbf{G} \varphi_1$) and 1-step reachability ($\llangle A \rrangle \textbf{X} \varphi_1$) objectives where no memory is needed for winning strategies \cite{EJ91}, it follows that all types of semantics considered are equal in $ATL$ with complete information as noted in \cite{Sch04}. Since $ATL^*_0$ can only be used to express $LTL$ objectives, it follows that $\models^{ATL^*}_{IF} \hspace{0.2cm} = \hspace{0.2cm} \models^{ATL^*}_{IR}$ since only finite memory is needed for winning strategies in such games \cite{PR89}.

\subsection{The bounded-memory hierarchy}

The bounded-memory hierarchy is increasing for $ATL_0/ATL_0^*$ because when a coalition has a $k$-memory winning strategy, then it also has a $k+1$-memory winning strategy which can be obtained by adding a disconnected memory-state to the DFST representing the strategy. For $ATL^*_0$ with complete information the hierarchy is strict. This can be seen since the family $\varphi_k = \llangle \{1 \} \rrangle \textbf{X}^k p$ of formulas for $k \ge 1$ has the property that $\mathcal{M}, s_0 \models_{IF_k} \varphi_k$ and $\mathcal{M},s_0 \not \models_{IF_{k-1}} \varphi_k$ for $k \ge 2$ for the one-player CGM $\mathcal{M}$ illustrated in Figure \ref{fig:m1}. Here player 1 wins if he chooses $w$ (wait) the first $k-1$ rounds and then chooses $g$ (go) in the $k$th round.

\begin{figure}[here]
\begin{center}
\begin{tikzpicture}

\tikzstyle{every node}=[ellipse, draw=black,
                        inner sep=0pt, minimum width=25pt, minimum height=25pt]

\draw (2,5) node [label=below:$s_0$] (s0)	{};
\draw (4,5) node [label=below:$s_1$] (s1)	{$p$};
\draw (6,5) node [label=below:$s_2$] (s2) {};

\path[->] (s0) edge node [below, draw=none] {$g$ } (s1);
\path[->] (s1) edge node [below, draw=none] {$w,g$ } (s2);

\path[->] (s0) edge [loop left] node [left, draw=none] {$w$ } (s0);
\path[->] (s2) edge [loop right] node [right, draw=none] {$g,w$ } (s2);

\end{tikzpicture}
\end{center}
\caption{CGM $\mathcal{M}$}
\label{fig:m1}

\end{figure}
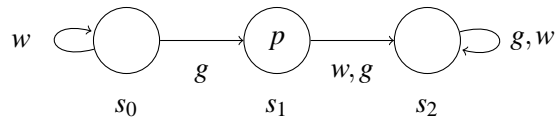

The reason that the property $\textbf{X}^k p$ cannot be forced by player 1 using a $(k-1)$-memory strategy is that the DFST representing the strategy would have to output the action $w$ in the first $k-1$ rounds followed by an output of the action $g$ when reading the same input $s_0$ in every round. This is not possible, because after $k-1$ rounds there must have been at least one repeated memory-state and from such a repeated state, the DFST would keep repeating its behavior. Therefore, it will either output $w$ forever or output $g$ before the $k$th round, making it unable to enforce $\textbf{X}^k p$. For $ATL_0/ATL^*_0$ with incomplete information, we can show the same result for the formula $\psi = \llangle \{1 \} \rrangle \textbf{F} p$ for the family $\mathcal{M}_k$ of iCGMs illustrated in Figure \ref{fig:m2} where $k \ge 1$. In this game all states except $s_0$ are in the same observation set for player 1. Here we have $\mathcal{M}_k, s_0 \models_{iF_k} \psi$ and $\mathcal{M}_k, s_0 \not \models_{iF_{k-1}} \psi$.

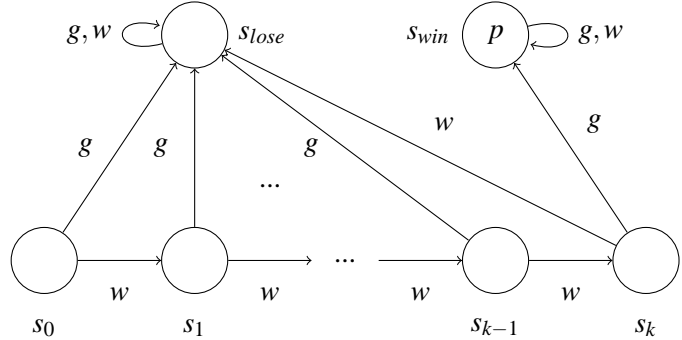
\begin{figure}[here]
\begin{center}
\begin{tikzpicture}

\tikzstyle{every node}=[ellipse, draw=black,
                        inner sep=0pt, minimum width=25pt, minimum height=25pt]

\draw (2,5) node [label=below:$s_0$] (s0)	{};
\draw (4,5) node [label=below:$s_1$] (s1)	{};
\draw (6,5) node [draw=none] (dots) {...};
\draw (8,5) node [label=below:$s_{k-1}$] (sk-1) {};
\draw (10,5) node [label=below:$s_{k}$] (sk) {};

\draw (4,8) node [label=right: $s_{lose}$] (sl)	{};
\draw (8,8) node [label=left: $s_{win}$] (sw)	{$p$};

\draw (5,6) node [draw=none] (dots2) {...};

\path[->] (s0) edge node [below, draw=none] {$w$ } (s1);
\path[->] (s1) edge node [below, draw=none] {$w$ } (dots);
\path[->] (dots) edge node [below, draw=none] {$w$ } (sk-1);
\path[->] (sk-1) edge node [below, draw=none] {$w$ } (sk);

\path[->] (s0) edge node [left, draw=none] {$g$ } (sl);
\path[->] (s1) edge node [left, draw=none] {$g$ } (sl);
\path[->] (sk-1) edge node [left, draw=none] {$g$ } (sl);
\path[->] (sk) edge node [above right, draw=none] {$w$ } (sl);
\path[->] (sk) edge node [above right, draw=none] {$g$ } (sw);

\path[->] (sl) edge [loop left] node [left, draw=none] {$g,w$ } (sl);
\path[->] (sw) edge [loop right] node [right, draw=none] {$g,w$ } (sw);

\end{tikzpicture}
\end{center}
\caption{iCGM $\mathcal{M}_k$}
\label{fig:m2}

\end{figure}

Player 1 wins exactly if he chooses $w$ for the first $k$ rounds and then $g$, which is not possible for a $(k-1)$-memory strategy when it receives the same input symbol in every round after the initial round as in the previous example.

The reason why the bounded-memory hierarchies are not increasing for $ATL/ATL^*$ in general is the possibility of using negation of strategically quantified formulas. For instance, given an $ATL_0$ formula $\varphi$, an iCGM $\mathcal{M}$ and a state $s$ such that $\mathcal{M}, s \models_{iF_k} \varphi$ and $\mathcal{M}, s \not\models_{iF_{k-1}} \varphi$ for some $k$, then for the $ATL$ formula $\neg \varphi$ we have $\mathcal{M}, s \not \models_{iF_k} \neg \varphi$ and $\mathcal{M}, s \models_{iF_{k-1}} \neg \varphi$.

\subsection{Infinite memory is needed}
\label{sec:expinf}

Finally, infinite memory is actually needed in some cases for $ATL_0/ATL^*_0$ with incomplete information. This is shown in a slightly different framework in \cite{BK10} where an example of a game is given with initial state $s_0$ such that $\mathcal{M}, s_0 \models_{iR} \llangle \{1,2\} \rrangle \textbf{G} \neg p$ and $\mathcal{M}, s_0 \not\models_{iF} \llangle \{1,2\} \rrangle \textbf{G} \neg p$ for a proposition $p$. We will not repeat the example here, but in the undecidability proof in Section \ref{sec:undec} an example of such a game is given. This means that $\models^L_{iF} \hspace{0.2cm} \neq \hspace{0.2cm} \models^L_{iR}$ for $L \in \{ATL_0, ATL^*_0\}$. We have $\models^{L}_{iF} \hspace{0.2cm}\subseteq\hspace{0.2cm} \models^L_{iR}$ since all finite-memory strategies are perfect recall strategies and therefore $\models^L_{iF} \hspace{0.2cm}\subset\hspace{0.2cm} \models^L_{iR}$ which concludes the last result of Figure \ref{fig:tabexpress}.

\section{Model-checking}

In this section we look at the decidability and complexity of model-checking $ATL/ATL^*$ with the new semantics introduced and compare with the results for memoryless and perfect recall semantics. We adopt the same way of measuring input size as in \cite{AHK97, AHK02, Sch04, JD08} where the input is measured as the size of the game structure and the size of the formula to be checked. In the case of bounded-memory semantics, we also include in the input size the size of the memory-bound $k$ encoded in unary. Our results can be seen in Figure \ref{tab:table1} along with known results for memoryless and perfect recall semantics.

\begin{figure}[here]
\begin{center}
\begin{tabular} { p{0.75cm} || p{5cm} | p{5cm}}

& $ATL$ & $ATL^*$\\
\hline
$\models_{Ir}$ & $PTIME$ \cite{AHK02} & $PSPACE$ \cite{Sch04}\\
$\models_{IF_k}$ & $PTIME$ & $PSPACE$\\
$\models_{IF}$ & $PTIME$ & $2EXPTIME$\\
$\models_{IR}$ & $PTIME$ \cite{AHK02} & $2EXPTIME$ \cite{AHK02}\\

\end{tabular}

\vspace{0.5cm}
\begin{tabular} { p{0.75cm} || p{5cm} | p{5cm}}

& $ATL$ & $ATL^*$\\
\hline
$\models_{ir}$ & $\Delta_2^p$ \cite{Sch04, JD08} & $PSPACE$ \cite{Sch04}\\
$\models_{iF_k}$ & $\Delta_2^p$ & $PSPACE$\\
$\models_{iF}$ & Undecidable & Undecidable\\
$\models_{iR}$ & Undecidable \cite{AHK02, DT11} & Undecidable \cite{AHK02, DT11}\\

\end{tabular}
\end{center}
\caption{Model-checking complexity for $ATL,ATL^*$. All complexity results are completeness results.}
\label{tab:table1}

\end{figure}

As can be seen in the figure, we obtain the same complexity for bounded-memory semantics as for memoryless semantics in all the cases which is positive, since we can solve many more games while staying in the same complexity class. We also obtain the same complexity for finite-memory semantics as perfect recall semantics, including undecidability for incomplete information games, which is disappointing. We will use the rest of the section to prove these results. In many cases this is done by using known results and techniques and modifying them slightly as well as using the results from Section \ref{sec:exp}.

\subsection{Using expressiveness results}

In section \ref{sec:exp} it was shown that $\models^{ATL}_{Ir} \hspace{0.2cm} = \hspace{0.2cm} \models^{ATL}_{IF_2} \hspace{0.2cm} = \hspace{0.2cm} \models^{ATL}_{IF_3} \hspace{0.2cm} = \hspace{0.2cm} ... \hspace{0.2cm} = \hspace{0.2cm} \models^{ATL}_{IF}$ which means that the model-checking problem is the same for these cases. Since $\models^{ATL}_{Ir}$ is known to be $PTIME$-complete \cite{AHK02} the result is the same for finite-memory semantics and bounded-memory semantics. It was also shown that $\models^{ATL^*}_{IF} \hspace{0.2cm} = \hspace{0.2cm} \models^{ATL^*}_{IR}$. Since model-checking $ATL_{IR}^*$ is $2EXPTIME$-complete \cite{AHK02} so is model-checking $ATL^*_{IF}$ since it is the same problem.

\subsection{Bounded-memory semantics}

For model-checking $ATL_{iF_k}, ATL^*_{IF_k}$ and $ATL^*_{iF_k}$ we employ some of the same ideas as in \cite{Sch04} for memoryless semantics, but extend them to deal with bounded-memory strategies. We first consider model-checking $ATL^*_0$ formulas with $iF_k$ semantics. Model-checking an $ATL^*_0$ formula $\llangle A \rrangle \varphi$ in an iCGM $\mathcal{M} = (\mathcal{G}, \pi)$ with $\mathcal{G} = (\states, \agt, \act, \mov, \tab, (\sim_j)_{1 \le j \le n})$ and initial state $s_0$ can be done using non-determinism as follows. First, assume without loss of generality that $A = \{1,...,r\}$ with $r \le n$. Use non-determinism to guess a $k$-memory strategy $\sigma = (\sigma_j)_{1 \le j \le r}$ for each of the players in $A$ represented by DFSTs $\mathcal{A}_j = (M_j, m_{j0}, {[]}_j, \act, T_j,G_j)$ for $j \in A$. Check that this strategy enforces $\varphi$ by creating a labelled and initialized transition system $T(s'_0,\sigma) = (Q,R,L,q_0)$ for all $s'_0 \sim_j s_0$ for some $1 \le j \le r$ in which the set of paths corresponds to the $\sigma$-outcomes from $s'_0$ in $\mathcal{G}$. The set $Q$ of states, the transition relation $R \subseteq Q \times Q$, the labeling function $L: Q \in 2^{\textup{Prop}}$ and the initial state $q_0$ are constructed as follows.

\begin{itemize}

\item $Q = \states \times \prod_{j = 1}^r M_j$

\item $([s,(m_1,...,m_r)], [s',(m'_1,...,m'_r)]) \in R$ if and only if there exists $a_{r+1},...,a_n \in \act$ so

\begin{itemize}

\item $\tab(s,(G_1(m_1,[s]_1),...,G_r(m_r,[s]_r),a_{r+1},...,a_n)) = s'$ and

\item $T_j(m_j,[s]_j) = m'_j$ for $1 \le j \le r$

\end{itemize}

\item $L(s,(m_1,...,m_r)) = \pi(s)$ for all $(s,(m_1,...,m_r)) \in Q$

\item $q_0 = (s'_0,(m_{10},...,m_{r0}))$

\end{itemize}

Intuitively, each state in the transition system corresponds to a state of the game as well as possible combinations of memory values for players in $A$. It can then be shown that $\rho = \rho_0 \rho_1 ...$ is a $\sigma$-outcome in $\mathcal{G}$ from $\rho_0 = s'_0$ if and only if there exists $(m_{1j},...,m_{rj}) \in \prod_{j = 1}^r M_r$ for $j \ge 0$ such that $\rho' = (\rho_0,(m_{10},...,m_{r0})) (\rho_1,(m_{11},...,m_{r1})) ...$ is a path in $T(s'_0,\sigma)$. This means that $\sigma$ is a witness that $\mathcal{M}, s_0 \models_{iF_k} \llangle A \rrangle \varphi$ if and only if $T(s'_0, \sigma),q_0 \models_{CTL^*} \textbf{A} \varphi$ for all $s'_0 \sim_j s_0$ for some $1 \le j \le r$. Note that the size of the transition systems are polynomial in the size of the input because $|Q| = k^r$, the number $n$ of agents is fixed and $r \le n$. In addition, the transition systems $T(s_1,\sigma)$ and $T(s_2,\sigma)$ are equal for any $s_1,s_2 \in \states$ except for the initial state of the transition systems. Thus, we can use the same transition system to do the check for the different initial states. We can perform this check of a strategy $\sigma$ in $PSPACE$ since $CTL^*$ model-checking can be done in $PSPACE$ \cite{CES86}. Moreover, when $\llangle A \rrangle \varphi$ is an $ATL_0$ formula, the check can be done in $PTIME$ since $CTL$ model-checking can be done in $PTIME$ \cite{CE81}. Thus, we can do model-checking of $ATL_0$ and $ATL^*_0$ with $iF_k$ semantics in $NP$ and $PSPACE$ respectively.

We extend the above algorithm to full $ATL$ and $ATL^*$ by evaluating the strategically quantified subformulas in a bottom up fashion, starting with the innermost formula and moving outwards resembling the technique typically used in $CTL^*$ model-checking \cite{CES86}. In both cases we need to make a linear amount of calls to the $ATL_0/ATL^*_0$ algorithm in the size of the formula to be checked. This gives us a $\Delta^p_2 = P^{NP}$ algorithm and a $PSPACE$ algorithm in $ATL$ and $ATL^*$ respectively. Since $ATL^*$ with $IF_k$ semantics is a special case, the $PSPACE$ algorithm also works here. The $PSPACE$-hardness for $ATL^*_{iF_k}$ and $ATL^*_{IF_k}$ follows from $PSPACE$-hardness of $ATL^*_{Ir}$ \cite{Sch04} since this is a special case of the two. In the same way $\Delta_2^p$-hardness of $ATL_{iF_k}$ follows from $\Delta_2^p$ hardness of $ATL_{ir}$ \cite{JD08}.

\subsection{Undecidability of finite-memory semantics}

\label{sec:undec}

In \cite{DT11} it was proven that model-checking $ATL$ and $ATL^*$ with $iR$ semantics is undecidable, even for as simple a formula as $\llangle A \rrangle \textbf{G} \neg p$ for $n \ge 3$ players. We provide a proof sketch for the same result for $iF$ semantics inspired by a technique from \cite{BK10} which also illustrates that infinite memory is needed in some games. The idea is to reduce the problem of whether a deterministic Turing machine with a semi-infinite tape that never writes the blank symbol repeats some configuration twice when started with an empty input tape, with the convention that the Turing machine will keep looping in a halting configuration forever if a halting state is reached. This problem is undecidable since the halting problem can be reduced to it. From a given Turing machine $T = (Q,q_0,\Sigma, \delta, B, F)$ of this type where $Q$ is the set of states, $q_0$ is the initial state, $\Sigma$ is the tape alphabet, $\delta: Q \times (\Sigma \cup \{B\}) \rightarrow Q \times \Sigma \times \{L,R\}$ is the transition function, $B$ is the blank symbol and $F$ is the set of accepting states, we generate a three-player concurrent game model $\mathcal{M}_T = (\mathcal{G}_T, \pi_T)$ with a state $s_0$ such that $\mathcal{M}_T, s_0 \models_{iF} \llangle \{ 1,2 \} \rrangle \textbf{G} \neg p$ if and only if $T$ repeats some configuration twice.

Consider the three-player game $\mathcal{M}_T$ in Figure \ref{fig:m4}. To make the figure more simple, we only write the actions of player 1 and 2 along edges and let player 3 choose a successor state, given the choices of player 1 and 2. If player 1 and 2 choose an action tuple that is not present on an edge from the current state of the game, the play goes to a sink state where $p$ is true. In all other states $p$ is false. Both player 1 and 2 have three observation sets, which are denoted $0$, $\cdot$ and I (though, they are not equal for the two players). In the figure we write $x \mid y$ in a state if the state is in observation set $x$ for player 1 and $y$ for player 2. The play starts in $s_0$ which is the only state in observation set $0$ for both player 1 and 2. The rules of the game are such that player 3 can choose when to let player 1 receive observation I. He can also choose to either let player 2 receive observation I at the same time as player 1 or let him receive it in the immediately following state of the game. Both player 1 and 2 can observe I at most once during the game. It can be seen from the game graph that both player 1 and 2 must play action $a$ until they receive observation I in order not to lose. We design the game so they must play the $v$th configuration of the Turing machine $T$ when receiving observation I after $v$ rounds in a winning strategy for all $v \ge 1$. To do this we let the tape alphabet and the set of control states of $T$ be legal actions for player 1 and 2. By playing a configuration, we mean playing the contents of the non-blank part of the tape of $T$ one symbol at a time from left to right and playing the control state immediately before the content of the cell that the tape head points to.

\begin{figure}[here]
\begin{center}
\begin{tikzpicture}

\tikzstyle{every node}=[ellipse, draw=black,
                        inner sep=0pt, minimum width=25pt, minimum height=25pt]

\draw (0,5) node [label=below:$s_0$] (s0)	{$0 \mid 0$};
\draw (2,7) node (s1)	{I $ \mid $ I};
\draw (4,7) node (s2)	{$\cdot  \mid \cdot$ };
\draw (6,5) node (s3)	{$\cdot  \mid \cdot$ };

%\draw (7,7) node [draw=none] (s2)	{$\begin{array}{l}\textup{Player 1 and 2 must play}\\\textup{initial config of } T\\ \end{array}$};
%\draw (6,7) node (s3)	{};
%\draw (6,3) node (s4)	{$\cdot \mid \cdot$};
\draw (2,4) node (s5)	{$\cdot \mid \cdot$};
\draw (4,5) node (s6)	{I $\mid$ I};
%\draw (9.5,5) node [draw=none] (s7)	{$\begin{array}{l}\textup{Player 1 and 2 must play}\\\textup{the same sequence of symbols }\\ \end{array}$};
%\draw (8,5) node (s8)	{$\cdot \mid \cdot$};
\draw (4,2) node (s9)	{I $\mid \cdot$ };
\draw (6,2) node (s10)	{$\cdot \mid$ I};
%\draw (8,1) node (s11)	{$\cdot \mid \cdot$};
%\draw (10,1) node (s12)	{$\cdot \mid \cdot$};
\draw (11,2) node [draw=none] (s7)	{$\begin{array}{l}\textup{Player 1 and 2 must play configs}  \\ C_1 \textup{ and } C_2 \textup{ such that } C_1 \vdash_T C_2 \\ \end{array}$};

\draw (3.7,7) node [minimum width = 150pt, minimum height = 50pt, style=dashed, rectangle, label = right:$\mathcal{M}_1$] (m1) {};

\draw (5.7,5.1) node [minimum width = 150pt, minimum height = 47pt, style=dashed, rectangle, label = right:$\mathcal{M}_2$] (m2) {};

\draw (8.5,2) node [minimum width = 320pt, minimum height = 50pt, style=dashed, rectangle, label = right:$\mathcal{M}_3$] (m3) {};

\path[->] (s0) edge node [above left, draw=none] {$(a,a)$ } (s1);
\path[->] (s0) edge [bend right] node [below left, draw=none] {$(a,a)$ } (s9);
\path[->] (s1) edge node [above, draw=none] {$(q_0,q_0)$} (s2);
\path[->] (s2) edge [loop right] node [right, draw=none] {$(*,*)$ } (s2);
\path[->] (s3) edge [loop right] node [right, draw=none] {$(*,*)$ } (s3);
%\path[->] (s2) edge node [above, draw=none] {$w$ } (s3);
%\path[->] (s2) edge [loop right] node [right, draw=none] {$(a,a)$ } (s2);
\path[->] (s0) edge node [above right, draw=none] {$(a,a)$ } (s5);
\path[->] (s5) edge [loop right] node [right, draw=none] {$(a,a)$ } (s5);
\path[->] (s5) edge node [above left = 2, draw=none] {$(a,a)$ } (s6);
\path[->] (s6) edge node [above, draw=none] {$(*,*)$} (s3);
\path[->] (s5) edge node [above right, draw=none] {$(a,a)$ } (s9);
%\path[->] (s4) edge node [above right, draw=none] {$(a,b)$ } (s11);
%\path[->] (s4) edge [loop right] node [above right, draw=none] {$(b,b)$ } (s4);
\path[->] (s9) edge node [below, draw=none] { } (s10);
\path[->] (s10) edge node [left, draw=none] {} (8,2);
%\path[->] (s10) edge node [above, draw=none] {$(a,b)$ } (s11);
%\path[->] (s11) edge node [above, draw=none] {$(a,b)$ } (s12);
%\path[->] (s12) edge [loop right] node [right, draw=none] {$(a,a)$ } (s12);
\end{tikzpicture}
\end{center}
\caption{iCGM $\mathcal{M}_T$}
\label{fig:m4}

\end{figure}
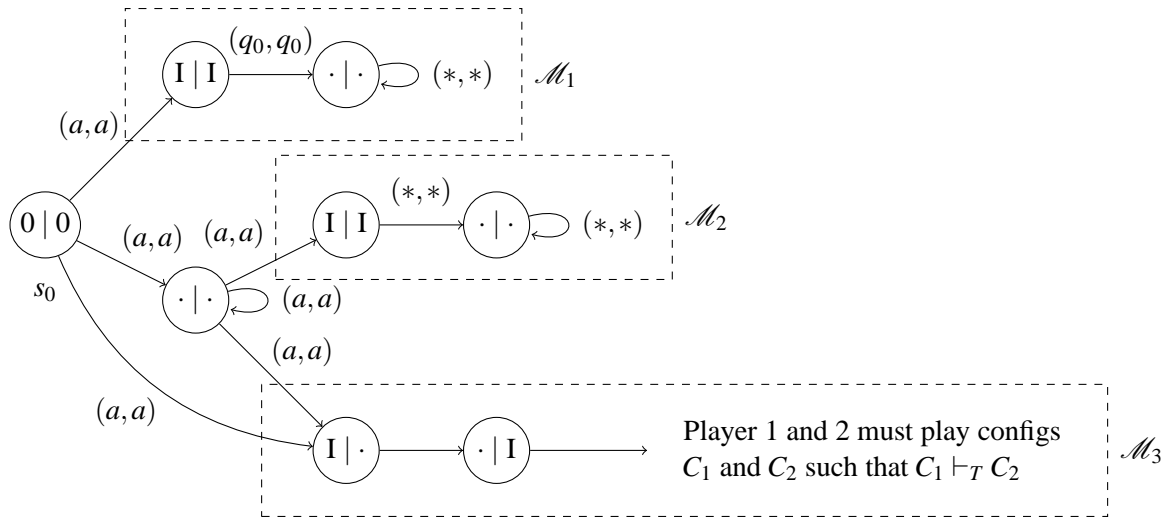

We design $\mathcal{M}_T$ with three modules $\mathcal{M}_1, \mathcal{M}_2$ and $\mathcal{M}_3$ as shown in Figure \ref{fig:m4}. They are designed with the following properties

\begin{itemize}

\item $\mathcal{M}_1$ is designed such that when player 1 and 2 both observe I after the first round, then in a winning strategy they must both play the initial configuration (i.e. $q_0$) in order to maintain $\neg p$. If they don't, then player 3 has a counter-strategy that takes the play to $\mathcal{M}_1$.

\item $\mathcal{M}_2$ is designed such that when player 1 and 2 both observe I at the same time, then in a winning strategy they must both play the same sequence of symbols after observing I ($*$ stands for any action and $(*,*)$ means any action pair where the two actions are equal). If there is a number $r>1$ so they don't comply with this when observing I after round $r$, then player 3 has a counter-strategy that takes the play to $\mathcal{M}_2$ after round $r$.

\item $\mathcal{M}_3$ is designed such that if player 1 observes I in the round before player 2 does, then in a winning strategy they must player configurations $C_1$ and $C_2$ respectively such that $C_1 \vdash_T C_2$ where $\vdash_T$ is the successor relation for configurations of $T$. Due to space limitations, the specific design of this module is omitted here.

\end{itemize}

Now, suppose $T$ has a repeated configuration. Then player 1 and 2 have a winning strategy $\sigma$ that consists in both players playing the $j$th configuration of the run of $T$ when observing I after the $j$th round. This strategy is winning because no matter if player 3 chooses to go to module $\mathcal{M}_1, \mathcal{M}_2, \mathcal{M}_3$  or none of them, then $\neg p$ will always hold given how they are designed when player 1 and 2 play according to $\sigma$. Next, the sequence of configurations in the run of $T$ is of the form $\pi \cdot \tau^\omega$ where $\pi$ and $\tau$ are finite sequences of configurations since $T$ has a repeated configuration. Then, player 1 and 2 only need finite memory to play according to $\sigma$ since they only need to remember a finite number of configurations and how far on the periodic path $\pi \cdot \tau^\omega$ the play is. Thus, they have a finite-memory winning strategy.

Suppose on the other hand that $T$ does not have a repeated configuration and assume for contradiction that player 1 and 2 have a $k$-memory winning strategy $\sigma$ for some $k$. Since player 1 and 2 cannot see whether the play is in $\mathcal{M}_1, \mathcal{M}_2$ or $\mathcal{M}_3$ player 1 must, when playing according to $\sigma$, play the first configuration $D_1$ of the run of $T$ when observing I after the first round. Otherwise, player 3 has a counter-strategy taking the play to $\mathcal{M}_1$ after the first round. Then, player 2 must play the second configuration $D_2$ of the run of $T$ when observing I after the second round. Otherwise, player 3 has a counter-strategy taking the play to $\mathcal{M}_3$ after the first round since player 1 must play $D_1$ when observing I after the first round and player 2 must play a successor configuration of what player 1 plays. Next, when using $\sigma$, player 1 must play $D_2$ when observing I after the second round. Otherwise, player 3 has a counter-strategy that takes the play to $\mathcal{M}_2$ after the second round since player 2 plays $D_2$ when observing I after the second round. Repeating this argument, it can be seen that $\sigma$ must consist of player 1 and 2 playing the $j$th configuration of the run of $T$ when observing I after the $j$th round for all $j \ge 1$. However, this is not possible for a $k$-memory strategy when the run of $T$ does not have a repeated configuration. This is because the current memory value of the DFST representing the strategy at the point when I is observed determines which sequence of symbols the strategy will play (since it will receive the same input symbol for the rest of the game). Thus, it is not capable of playing more than $k$ different configurations. And since for any $k$ a winning strategy must be able to play more than $k$ different configurations there is a contradiction and a finite-memory winning strategy therefore cannot exist.

In conclusion $\mathcal{M}_T, s_0 \models_{iF} \llangle \{ 1,2 \} \rrangle \textbf{G} \neg p$ if and only if $T$ repeats some configuration twice, which means that the model-checking problem is undecidable for $ATL$ and $ATL^*$ with $iF$ semantics. This game also illustrates that infinite memory is needed in some games, since player 1 and 2 can win the game with perfect recall strategies when $T$ does not have a repeated configuration. This is simply done by playing the sequence of configurations of the run of $T$.

\section{Concluding Remarks}

We have motivated the extension of the alternating-time temporal logics $ATL/ATL^*$ with bounded-memory and finite-memory semantics and have explored the expressiveness for both complete and incomplete information games. Both finite-memory semantics and the infinite hierarchy of bounded-memory semantics were shown to be different from memoryless and perfect recall semantics. We have also obtained complexity and decidability results for the model-checking problems that emerged from the newly introduced semantics. In particular, the model-checking results for bounded-memory semantics were positive with as low a complexity as for memoryless semantics for $ATL/ATL^*$ and complete/incomplete information games. Unfortunately model-checking with finite-memory semantics was shown to be as hard as with perfect recall semantics in the cases considered, even though it was shown to be a different problem.

%\nocite{*}
\bibliographystyle{eptcs}
\bibliography{generic}
\end{document}